\documentclass[pra,twocolumn,aps,preprintnumbers,showpacs,superscriptaddress,10pt]{revtex4-1}
\usepackage{hyperref}
\usepackage{verbatim}
\usepackage{amsmath}
\usepackage{latexsym}

\usepackage{natbib}
\usepackage{amsfonts}
\usepackage{amsmath}
\usepackage{amssymb}
\usepackage{amsthm}
\usepackage{graphicx}
\usepackage{bm}

\newcommand{\be}{\begin{eqnarray} \begin{aligned}}
\newcommand{\ee}{\end{aligned} \end{eqnarray} }
\newcommand{\benn}{\begin{eqnarray*} \begin{aligned}}
\newcommand{\eenn}{\end{aligned} \end{eqnarray*} }

\newcommand{\bc}{\begin{center}}
\newcommand{\ec}{\end{center}}
\newcommand{\half}{\frac{1}{2}}

\newcommand{\id}{\mathbb{I}}

\newcommand{\tr}{\mathop{\mathrm{tr}}\nolimits}

\newtheorem{theorem}{Theorem}[section]

\newtheorem{lemma}[theorem]{Lemma}

\newcommand{\ra}{\rightarrow}
\newcommand{\hil}{\mathcal{H}}

\newcommand{\nn}{\nonumber}

\usepackage{amsfonts}
\def\Real{\mathbb{R}}
\def\Complex{\mathbb{C}}

\def\id{\mathbb{I}}

\def\01{\{0,1\}}

\newcommand{\eps}{\varepsilon}
\newcommand{\ket}[1]{|#1\rangle}
\newcommand{\bra}[1]{\langle#1|}

\newcommand{\proj}[1]{|#1\rangle\langle#1|}
\newcommand{\ketbra}[2]{|#1\rangle\langle#2|}

\newcommand{\ignore}[1]{}

\newcommand{\vN}{{\ensuremath{{\rm H}}}}
\newcommand{\hmin}{{\ensuremath{{\rm H}}_{\min}}}
\newcommand{\hmineps}{{\ensuremath{{\rm H}}_{\min}^\eps}}
\newcommand{\hminarg}[1]{{\ensuremath{{\rm H}}_{\min}^{#1}}}
\newcommand{\hmax}{{\ensuremath{{\rm H}}_{\max}}}
\newcommand{\hmaxeps}{{\ensuremath{{\rm H}}_{\max}^\eps}}
\newcommand{\hmaxarg}[1]{{\ensuremath{{\rm H}}_{\max}^{#1}}}

\newcommand{\sys}{{\ensuremath{ S } }}
\newcommand{\env}{{\ensuremath{ E } }}
\newcommand{\ssys}{_{\sys}}
\newcommand{\senv}{_{\env}}

\newcommand{\hsys}{\mathcal{H}_{\ensuremath{ S } }}
\newcommand{\henv}{\mathcal{H}_{\ensuremath{ E } }}
\newcommand{\omg}{\mathcal{H}_{\Omega_E}}
\newcommand{\somg}{\mathcal{H}_{\Omega_S}}

\newcommand{\cS}{\ensuremath{\mathcal{S}}}
\newcommand{\cT}{\ensuremath{\mathcal{T}}}

\newcommand{\cI}{\ensuremath{\mathcal{I}}}
\newcommand{\mdag}{^{\dag}}
\newcommand{\lmax}{\lambda_{\max}}
\newcommand{\cj}{Choi-Jamio\l{}kowski }
\DeclareMathOperator{\Herm}{Herm}

\newcommand{\trdist}[2]{\left\|#1 - #2\right\|_1}

\newcommand{\UU}{\mathbb{U}}

\begin{document}

\title{Dependence of a quantum mechanical system on its own initial state and the initial state of the environment it interacts with}
\author{Adrian Hutter}\email{adrian.hutter@unibas.ch}
\affiliation{Centre for Quantum Technologies, National University of Singapore, 2 Science Drive 3, 117543 Singapore}
\affiliation{Department of Physics, University of Basel, Klingelbergstrasse 82, CH-4056 Basel, Switzerland}
\author{Stephanie Wehner}
\affiliation{Centre for Quantum Technologies, National University of Singapore, 2 Science Drive 3, 117543 Singapore}
\email{wehner@nus.edu.sg}
\date{\today}

\begin{abstract}
	We present a unifying framework to the understanding of when and how quantum mechanical systems become independent 
	of their initial conditions and adapt macroscopic properties (like temperature) of the environment.
	By viewing this problem from an quantum information theory
	perspective, we are able to simplify it in a very natural and easy way. We first 
	show that for \emph{any} interaction between the system
	and the environment, and almost all initial states of the system, the question of how long the system retains memory of its initial conditions
	can be answered by studying the temporal evolution of just \emph{one} special initial state. This special state thereby depends only
	on our knowledge of macroscopic parameters of the system. We provide a simple entropic inequality for this state that can be used
	to determine whether mosts states of the system have, or have not become independent of their initial conditions after time $t$. 
	We discuss applications of our entropic criterion to thermalization times in systems with an effective light-cone and to quantum memories suffering depolarizing noise.
	We make a similar statement for almost all initial states of the environment, and finally provide a sufficient condition for which 
	a system never thermalizes, but remains close to its initial state for all times.
\end{abstract}
\maketitle

\section{Introduction}

We are all familiar with thermalization on a \emph{macroscopic} level - simply consider what happens when you leave your cup of coffee untouched for a while. Yet, understanding this process from a 
\emph{microscopic} level forms a challenging endeavour. How could we hope to justify thermalization from the rules of quantum mechanics?

To tackle this problem it is helpful to break it up into smaller, more manageable, components. As~\cite{linden:evolution} point out,
the straighforward-looking process of thermalization actually consists of four aspects which may be addressed independently. 
Roughly speaking, they deal with several different questions that we might ask about a system $\sys$ after it is placed into contact with an environment (bath) $\env$. 
The first of these is whether the system equilibrates
i.e.\ evolves towards some particular equilibrium state and remains close to it.
Note that when we only ask about equilibration, we do not care what form this equilibrium state actually takes. In particular, it may depend 
on the initial state of the system and/or the environment and does not need to be a thermal state. A second question is thus whether this equilibrium state \emph{is} indeed independent
of the intial state of the system. Note that one may also think of this question as asking whether the system retains at least some amount of memory of its precise initial conditions in equilibrium.
Similarly, the third question asks whether the equilibrium state depends on the precise details of the intial state of the environment, or only on its macroscopic parameters such as temperature.
Finally, if we find that the equilibrium state of the system is indeed independent of such initial states, we may then ask whether it actually takes on the familiar Boltzman form.

In this work, we focus on the second and third of these aspects, i.e.\ we are concerned with the following questions: 
\begin{itemize}
	\item \emph{Independence of the initial state of the system (environment).} At time $t$, does the state of the system depend on the precise initial state of the system (environment)? (or only on its macroscopic parameters?)
\end{itemize}
We present a unifying framework to deal with problems related to these two questions. 
Our approach allows to make statements about the the evolved state of $S$ at any particular time $t$ and hence about \emph{time-scales} needed for thermalization. This is in contrast to previous references dealing with thermalization in a highly general setting that make statements about \emph{temporal averages} 
\cite{linden:evolution,linden:fluctuations,gogolin:mthesis,gogolin:absence,gogolin:thermalization}.

Our approach is based on an entropic condition that allows to decide whether at time $t$ almost all initial states from some subspace $\Omega_S$ of $S$ have evolved close to the same state $\rho_S(t)$ or not. 
Let $\ket{\psi}_E$ denote the initial state of the environment $E$ the system interacts with and let $\pi_{\Omega_S}$ denote the maximally mixed state on $\Omega_S$. From time $0$ to $t$, $S$ and $E$ undergo a joint unitary evolution $U(t)$.
We define the state
\begin{align}
	\tau_{SE}(t) = U(t) \left( \pi_{\Omega_S} \otimes \proj{\psi}_E \right) U(t)^\dagger\ .
\end{align}
Note that the state $\tau_{SE}(0)$ has zero entropy in $E$, while its entropy in $S$ is determined by the subspace dimension of $\Omega_S$. 
We predict the following: as long as the (smooth min-)entropy of the state $\tau_{SE}(t)$ in $S$ is larger than the (smooth max-)entropy in $E$,
\begin{align}
  \hmineps(S)_{\tau_{SE}(t)} \gtrsim \hmaxeps(E)_{\tau_{SE}(t)}\ ,
\end{align}
different initial states from $\Omega_S$ have not (yet) evolved to the same state; the system still ``remembers'' its initial state.
Conversely, if at any time $t$ the (smooth min-)entropy of the state $\tau_{SE}(t)$ in $E$ has become larger than the (smooth max-)entropy in $S$, 
\begin{align}
  \hmaxeps(S)_{\tau_{SE}(t)} \lesssim \hmineps(E)_{\tau_{SE}(t)}\ ,
\end{align}
almost all initials states from $\Omega_S$ will have evolved close to the same state, namely $\tau_{S}(t)$.
This condition is essentially tight up to differences between smooth min- and max-entropies. 
These are introduced in Sec.~\ref{sec:entropies}, together with our main tool from quantum information theory. 
Our entropic condition for independence of the initial state of the environment will be formally derived and stated in Sec.~\ref{sec:systemState}, 
while an analogous condition for independence of the initial state of the environment shall be presented in Sec.~\ref{sec:environmentState}.
For both conditions, we present two example applications. We conclude in Sec.~\ref{sec:discussion}.

\section{Entropy measures and quantum channel analysis}\label{sec:entropies}

While the relevant entropy in i.i.d.\ (\emph{independent and identically distributed}) scenarios is the well-known von Neumann entropy, the relevant quantities for a single experiment (a.k.a.\ \emph{single shot}) 
are the min- and max-entropies, well established in quantum information theory. For a bipartite system $AB$ these are defined as
\begin{align}\label{hmin}
\hmin(A|B)_\rho := \sup_{\sigma_B}\sup\left\{\lambda\in\Real:2^{-\lambda}\id_A\otimes\sigma_B\geq\rho_{AB}\right\}
\end{align}
and
\begin{align}\label{hmax}
\hmax(A|B)_\rho := \sup_{\sigma_B}\log \left[F (\rho_{AB},\id_A\otimes\sigma_B)\right]^2\ ,
\end{align}
where the suprema are over all density operators on the Hilbert space $\hil_B$ and $F(\rho ,\sigma) := \left\|\sqrt{\rho} \sqrt{\sigma}\right\|_1$ with $\|A\|_1 = \tr\sqrt{A^\dagger A}$ denotes the fidelity.
With $\log$ we always denote the binary logarithm.
We have from \cite[Lemma 2 and Lemma 20]{marco:fqaep} that
\begin{align}\label{hminmaxineq}
&-\log \min\left\{d_A, d_B\right\}\leq \hmin(A|B)_\rho\leq \vN(A|B)_\rho \nn\\&\quad\qquad \leq\hmax(A|B)_\rho\leq\log d_A\ .
\end{align}
For a single system $A$, these entropy measures can easily be expressed in terms of the eigenvalues $\{\lambda_j\}_j$ of the state $\rho_A = \sum_j \lambda_j \proj{j}$ as 
$\hmin(A)_{\rho} = - \log \max_j \lambda_j$ and $\hmax(A)_{\rho} = 2\log \sum_j \sqrt{\lambda_j}$. Both quantities enjoy nice operational interpretations in quantum information~\cite{renato:operational}
as well as thermodynamics~\cite{lidia:thermodynamic,lidia:decouple}.
We will also refer to \emph{smoothed} versions of these quantities $\hmineps$ and $\hmaxeps$ which can be thought of as equal to the original quantity, except up to an error $\eps$. 
Specifically, $\hmineps(A|B)_\rho$ ($\hmaxeps(A|B)_\rho$) is the maximum (minimum) $\hmin(A|B)_\sigma$ ($\hmax(A|B)_\sigma$) over all states $\sigma_{AB}$ which are $\eps$-close to $\rho_{AB}$.
The appropriate distance measure is thereby the \emph{purified distance} \cite{tomamichel:duality,definition_comment}.
On the other hand, when we say that two quantum states $\rho$ and $\sigma$ are \emph{close}, we mean that their trace distance $\|\rho - \sigma\|_1$ is very small \cite{helstrom_comment}.

Both entropy measures converge to the von Neumann entropy ${\rm H}(A|B)_\rho$ in the asymptotic limit of many i.i.d.\ experiments~\cite{marco:fqaep}.
Since the smooth entropies are invariant under local isometries \cite{tomamichel:duality}, we have for a pure state $\ket{\psi}_{AB}$ that
\begin{align}\label{eq:pureSym}
  \ensuremath{{\rm H}}_{\min/\max}^\eps(A)_{\proj{\psi}} = \ensuremath{{\rm H}}_{\min/\max}^\eps(B)_{\proj{\psi}}\ ,
\end{align}
which we shall use repeatedly. In order to lower-bound $\hmineps(A|B)_\rho$, we employ 
\begin{align}\label{minchain1eqn}
\hmineps(A|B)_\rho \geq \hmineps (AB)_\rho - \log d_B\ .
\end{align}
(which follows directly from \cite[Lemma 3.1.10.]{renato:diss} and the definition of $\hmineps$)
and \cite{vitanov:chainrules}
\begin{align}\label{minchain2eqn}
\hmineps (A|B)_\rho \geq \hminarg{\frac{\eps}{4}}(AB)_\rho - \hmaxarg{\frac{\eps}{4}}(B)_\rho - O\left(\log\frac{1}{\eps}\right)\ .
\end{align} 

Consider a \emph{quantum channel} $\cT_{A\ra B}$, which mathematically is a completely positive and trace-preserving mapping (henceforth CPTPM). 
Let $\tau_{A'B}$ be the \emph{\cj} representation of $\cT_{A\ra B}$, i.e.\ the state
\begin{align}
	\tau_{A'B} = (\cI_A \otimes \cT_{A\rightarrow B})\proj{\Psi}_{A'A}\ ,
\end{align}
where $A'$ is a copy of $A$ and $\ket{\Psi}_{AA'} := \frac{1}{\sqrt{d_A}}\sum_{i=1}^{d_A}\ket{i}_A\otimes\ket{i}_{A'}$ denotes the maximally entangled state across $A$ and $A'$.

The following theorem provides essentially tight entropic conditions for whether a quantum channel $\cT_{A\ra B}$ is such that most input states on $\hil_A$ yield the same output on $B$ -- or not.
It will be our main tool to investigate the quantum channels we are interested in for physical reasons.
\begin{theorem}\label{thm:main}
 With the above notation, we have
\begin{align}\label{eq:smallAvg}
 &\left\langle\left\|\cT_{A\ra B}(\proj{\phi}_A)-\cT_{A\ra B} \left(\pi_A\right)\right\|_1\right\rangle_{\ket{\phi}_A} \nn\\&\quad \leq 2^{-\half \hmineps (A'|B)_\tau} + O\left(\eps\right)\ ,
\end{align}
where $\left\langle\ldots\right\rangle_{\ket{\phi}_A}$ denotes the average over the uniform (Haar) measure on $\hil_A$ and $\pi_A$ denotes the maximally mixed state on $A$.
Furthermore, for each $\delta>0$
\begin{align}\label{eq:expStrong}
&\Pr_{\ket{\phi}_A}\left\{\left\|\cT_{A\ra B}(\proj{\phi}_A)-\cT_{A\ra B} \left(\pi_A\right)\right\|_1 \right.\nn\\&\left.\qquad   \geq
 2^{-\half \hmineps (A'|B)_\tau} + O\left(\eps\right) + \delta\right\} \nn\\&
\leq 2 e^{-d_A\delta^2/16}\ ,
\end{align}
where $\Pr_{\ket{\phi}_A}$ denotes the probability if $\ket{\phi}_A$ is picked at random from the Haar measure on $\hil_A$.
Conversely, if
\begin{align}\label{eq:entropyCondition}
 \hmaxeps (A'B)_\tau + \log\frac{1}{1-(\sqrt{2\delta}+4\eps)^2}+\log\frac{2}{\eps^2} < \hmineps (B)_\tau
\end{align}
there is \emph{no} state $\omega_B$ such that 
\begin{align}\label{eq:converse}
\left\langle \left\| \cT(\proj{\phi}_A) - \omega_B\right\|_1 \right\rangle_{\ket{\phi}_A} \leq \frac{\delta}{2}\ .
\end{align}
\end{theorem}
The first assertion, Eq.~\eqref{eq:smallAvg}, is a direct consequence of the decoupling theorem of \cite{dupuis:diss,dupuis:decoupling} with a trivial reference system $R$ (note that $\tau_B = \tr_{A'}\cT_{A\rightarrow B}\proj{\Psi}_{A'A} = \cT_{A\ra B} \left(\pi_A\right)$).
Applying the measure concentration properties of the Haar measure (see the proof of \cite[Theorem 3.9.]{dupuis:diss} and references therein) then implies \eqref{eq:expStrong}.
In order to make \eqref{eq:expStrong} strong, we need $d_A$ to be sufficiently large. We can then choose $\delta=d_A^{-1/3}$ in order to make both $\delta$ as well as the probability $2 e^{-d_A\delta^2/16}$ small.
The reason for this is that the measure concentration properties of the Haar measure only give strong results in high-dimensional spaces.
Note, however, that Hilbert space dimensions grow exponentially with the number of constituent particles, so in usual situations of physical interest Hilbert space dimensions will be huge.

If $\hmineps (A'|B)_\tau$ is sufficiently positive and $d_A$ is sufficiently large, almost almost all input states on $A$ yield the same channel output on $B$.
From \eqref{minchain2eqn} we have $\hmineps (A'|B)_\tau \gtrsim \hmineps_\tau(A'B) - \hmaxeps(B)_\tau$. We will thus state the condition for almost all input states on $A$ to yield the same output on $B$ slightly informally as
\begin{align}\label{eq:cond1}
 \hmineps(A'B)_\tau \gtrsim \hmaxeps(B)_\tau\ .
\end{align}
Conversely, \eqref{eq:entropyCondition} tells us that if 
\begin{align}\label{eq:cond2}
 \hmaxeps(A'B)_\tau \lesssim \hmineps(B)_\tau\ ,
\end{align}
there is a considerable chance that different input states on $A$ yield different outputs on $B$.
Note that unlike the entropic terms, the logarithmic terms in \eqref{minchain2eqn} and \eqref{eq:entropyCondition}, which we neglected in our informal conditions, do not grow with system size and are thus negligible for large systems.
Our entropic conditions are thus essentially tight up to differences between smooth min- and max-entropies. This difference vanishes in the asymptotic limit of many i.i.d.\ experiments~\cite{marco:fqaep}.
Consider a product input space $A^{\otimes n}$ and an i.i.d.\  channel $\cT_{A\rightarrow B}^{\otimes n}$.
The \cj representation of the product channel $\cT_{A\rightarrow B}^{\otimes n}$ takes the form $\tau_{A'B}^{\otimes n}$, where $\tau_{A'B}$ is the \cj representation of each single channel $\cT_{A\rightarrow B}$.
In order to decide whether most input staes from $A^{\otimes n}$ are mapped to the same state on $B^{\otimes n}$, the relevant entropic quantities are then according to Theorem \ref{thm:main} 
$\hmineps (A'|B)_{\tau_{A'B}^{\otimes n}}$ and $\hmaxeps (A'B)_{\tau_{A'B}^{\otimes n}}-\hmineps (B)_{\tau_{A'B}^{\otimes n}}$. We have from \cite{marco:fqaep} that 
\begin{align}\label{eq:fqaep}
&\lim_{\eps\rightarrow0}\lim_{n\rightarrow\infty}\frac{1}{n}\hmineps (A'|B)_{\tau_{A'B}^{\otimes n}} \nn\\
&=\lim_{\eps\rightarrow0}\lim_{n\rightarrow\infty}\frac{1}{n}\left(\hmaxeps (A'B)_{\tau_{A'B}^{\otimes n}}-\hmineps (B)_{\tau_{A'B}^{\otimes n}}\right) \nn\\
&=\vN(A'|B)_{\tau_{A'B}}\ ,
\end{align}
where $\vN(A'|B)_{\tau_{A'B}}$ denotes the conditional von Neumann entropy evaluated for the state $\tau_{A'B}$.
The sign of $\vN(A'|B)_{\tau_{A'B}}$ thus indeed provides a tight criterion in the limit of large $n$.
Note that for the parameter $\delta$ in \eqref{eq:entropyCondition} and \eqref{eq:converse} we may choose any value for which the logarithmic term on the l.h.s.\ of \eqref{eq:entropyCondition} is still well-defined.
In the asymptotic limit $n\rightarrow\infty$ where we choose $\eps\rightarrow0$ for \eqref{eq:fqaep}, \eqref{eq:converse} is thus valid with $\delta=\frac{1}{2}$, if $\vN(A'|B)_{\tau_{A'B}}<0$.

In order to prove the converse part of Theorem~\ref{thm:main}, we shall need the following auxiliary lemma.

\begin{lemma}\label{lem:hminConditional}
Let $\left\lbrace\ket{i}_R\right\rbrace_{i=1,\ldots,n}$ be an orthonormal family of stats in $\hil_R$ and let $\left\lbrace\ket{\psi(i)}_A\right\rbrace_{i=1,\ldots,n}$ be an arbitrary family of stats in $\hil_A$.
 For $\rho_{AR}=\frac{1}{n}\sum_{i=1}^n\proj{\psi(i)}_A\otimes\proj{i}_R$ we have
\begin{align}
 \hmineps(A|R)_\rho = \log\frac{1}{1-\eps^2}\ .
\end{align}
\end{lemma}
\begin{proof}
We use the notation $\cS_{=} (\hil_A):=\left\{\rho_A\in\Herm (\hil_A): \rho_A \geq 0, \tr \rho_A =1\right\}$ and $\cS_{\leq} (\hil_A):=\left\{\rho_A\in\Herm (\hil_A): \rho_A \geq 0, \tr \rho_A \leq1\right\}$.

Recall that $\hmineps(A|R)_\rho$ is defined as the  supremum of $\hmin(A|R)_\sigma$ over all states $\sigma_{AR}\in\cS_{\leq} (\hil_A\otimes\hil_R)$ with $P(\rho_{AR},\sigma_{AR})\leq\eps$, where $P$ denotes the purified distance \cite{tomamichel:duality}.
Since $\hil_A\otimes\hil_R$ is finite-dimensional, there is a state $\sigma_{AR}\in\cS_{\leq} (\hil_A\otimes\hil_R)$ achieving the supremum.
Since $\rho_{AR}$ is classical on $R$, the supremum can be restricted to states $\sigma_{AR}$ which are classical on $R$ as well (see \cite[Remark 3.2.5]{renato:diss}).
Consequently, there is a state $\sigma_{AR}=\sum_i\sigma_A^{(i)}\otimes\proj{i}_R$ with $\sigma_A^{(i)}\in\cS_{\leq} (\hil_A)$ such that $P(\rho_{AR},\sigma_{AR})\leq\eps$ and $\hmineps(A|R)_\rho=\hmin(A|R)_\sigma$.
From the definition of the purified distance \cite{tomamichel:duality}, we have 
\begin{align}
 P(\rho_{AR},\sigma_{AR}) &= \sqrt{1-\left(\tr\sqrt{\sqrt{\rho}\sigma\sqrt{\rho}}\right)^2} \nn\\&= \sqrt{1-\frac{1}{n}\left(\sum_{i=1}^n\sqrt{\langle\psi(i)\vert\sigma_A^{(i)}\vert\psi(i)\rangle}\right)^2}
\end{align}
and from the definition of the conditional min-entropy
\begin{align}
 \hmin(A|R)_\sigma =-\log\sum_{i=1}^n\lmax(\sigma_A^{(i)})\ . 
\end{align}
For a fixed $P(\rho_{AR},\sigma_{AR})$, the entropy $\hmin(A|R)_\sigma$ becomes maximal if we choose $\sigma_A^{(i)}=\mu(i)\proj{\psi(i)}_A$ with $\mu(i)\geq0$.
Given the constraint 
\begin{align}
 P(\rho_{AR},\sigma_{AR})&= \sqrt{1-\frac{1}{n}\left(\sum_{i=1}^n\sqrt{\mu(i)}\right)^2}
\nn\\&\leq\eps
\end{align}
the min-entropy $\hmin(A|R)_\sigma=-\log\sum_i\mu(i)$ then becomes maximal if we choose $\mu(i)=\frac{1-\eps^2}{n}$ for each $i$ and thus
\begin{align}
 \hmineps(A|R)_\rho = \hmin(A|R)_\sigma = \log\frac{1}{1-\eps^2}\ .
\end{align}
\end{proof}

With this lemma at hand, let us now proof the converse part of Theorem~\ref{thm:main}.

\begin{proof}
The proof consists of two parts. First we show that
\begin{align}\label{specificconverse}
\left\langle \left\| \cT(\proj{\phi}_A) - \cT(\pi_A)\right\|_1 \right\rangle_{\ket{\phi}_A} > \delta\ .
\end{align}
Then we show that if this is true the average cannot be small for any state $\omega_B\in\cS_{=} (\hil_B)$ (see the notation introduced in the proof of Lemma~\ref{lem:hminConditional}.

From \cite[Theorem 4.1]{dupuis:decoupling} we have that if for $\rho_{AR}\in\cS_{=}(\hil_{AR})$ and $\tilde{\tau}_{A'B} = d_A \sqrt{\rho_{A'}}\tau_{A'B}\sqrt{\rho_{A'}}$ the entropic condition 
\begin{align}\label{eq:converseCondition}
&\hminarg{4\eps+\sqrt{2\delta}}(A|R)_{\rho}+\hmaxeps(A'B)_{\tilde{\tau}}-\hmineps(B)_{\tilde{\tau}} 
< -\log\frac{2}{ {\eps}^2}\ .
\end{align}
is fulfilled (for arbitrary $\eps, \delta>0$), then
\begin{align}\label{eq:largeDist}
\left\| \cT(\rho_{AR}) - \cT(\rho_A) \otimes \rho_R \right\|_1 > \delta\ .
\end{align}
We apply the above result with
\begin{align}
\rho_{AR} := \int_{\UU(A)} U \proj{\phi}_{A} U\mdag \otimes \ket{U}\bra{U}_R dU\ ,
\end{align}
where the integral is over all unitaries $U$ from the Haar measure on the group of unitaries $\UU(A)$ on $\hil_A$.
We think of $R$ as being a classical register which holds the information about wich unitary $U$ has been applied.
Note that the entropy in Lemma~\ref{lem:hminConditional} is independent of $n$. We may thus consider a continuum limit $\frac{1}{n}\sum_{i=1}^n\rightarrow\int_{\UU(A)}dU$ and conclude that 
\begin{align}
 \hminarg{4\eps+\sqrt{2\delta}}(A|R)_{\rho} = \log\frac{1}{1-\left(4\eps+\sqrt{2\delta}\right)^2}
\end{align}

Since $\rho_A = \int_{\UU(A)} U \proj{\phi}_A U\mdag dU = \pi_A$ we have $\tilde{\tau}_{A'B} = d_A\sqrt{\rho_A}\tau_{A'B}\sqrt{\rho_A} = \tau_{A'B}$. 
The assumption \eqref{eq:converseCondition} is thus fulfilled if the assumption of the converse part of Theorem~\ref{thm:main} is fulfilled.
From \eqref{eq:largeDist} we have then that 
\begin{align}
\nn \delta &< \left\| \cT(\rho_{AR})-\cT(\rho_{A})\otimes\rho_{R}\right\|_{1} \nn\\
\nn &= \left\|\int_{\UU(A)} \cT(U \proj{\phi}_A U\mdag) \otimes \ket{U}\bra{U}_R dU
\right.\nn\\&\qquad\qquad\left. -  \cT(\pi_A) \otimes \int_{\UU(A)} \ket{U}\bra{U}_R dU \right\|_{1} \nn\\
\nn &= \left\|\int_{\UU(A)} \left\{\cT(U \proj{\phi}_A U\mdag) -  \cT(\pi_A)\right\} \otimes \ket{U}\bra{U}_R dU \right\|_{1} \nn\\
\nn &= \int_{\UU(A)}\left\|\left\{\cT(U \proj{\phi}_A U\mdag) -  \cT(\pi_A)\right\} \otimes \ket{U}\bra{U}_R \right\|_{1} dU \nn\\
&= \int_{\UU(A)}\left\|\cT(U \proj{\phi}_A U\mdag) -  \cT(\pi_A) \right\|_{1} dU \nn\\
&= \left\langle \left\| \cT(\proj{\phi}_A) - \cT(\pi_A)\right\|_1 \right\rangle_{\ket{\phi}_A}\ .
\end{align}
The third equality is due to the fact that all operators in the integral act on mutually orthogonal states due to the $R$-factor.

Now, assume by contradiction that there is a state $\omega_B \in \cS_{=}(\hil_B)$ such that 
\begin{align}\label{convpf0}
\left\langle \left\| \cT(\proj{\phi}_A) - \omega_B\right\|_1 \right\rangle_{\ket{\phi}_A} \leq \frac{\delta}{2}\ .
\end{align}
Then, by use of the triangle inequality,
\begin{align}\label{convpf1}
\frac{\delta}{2} 
&\geq \left\langle \left\| \cT(\proj{\phi}_A) - \cT(\pi_A)\right\|_1 \right\rangle_{\ket{\phi}_A} 
- \left\langle \left\| \cT(\pi_A) - \omega_B\right\|_1 \right\rangle_{\ket{\phi}_A} \nn\\
&> \delta - \left\|\cT(\pi_A) -  \omega_B\right\|_{1}\ .
\end{align}
Furthermore, by use of the convexity of the trace distance,
\begin{align}\label{convpf2}
&\left\langle \left\| \cT(\proj{\phi}_A) - \omega_B\right\|_1 \right\rangle_{\ket{\phi}_A} \nn\\
&\quad\geq \left\|\left\langle\cT(\proj{\phi}_A)\right\rangle_{\ket{\phi}_A} - \omega_B\right\|_1 \nn\\
&\quad= \left\|\cT(\pi_A) -  \omega_B \right\|_{1}\ .
\end{align}
Combining inequalities (\ref{convpf1}) and (\ref{convpf2}) yields
\begin{align}
\left\langle \left\| \cT(\proj{\phi}_A) - \omega_B\right\|_1 \right\rangle_{\ket{\phi}_A} > \frac{\delta}{2}
\end{align}
in contradiction to (\ref{convpf0}).
\end{proof}

\section{Independence of the initial state of the system}\label{sec:systemState} 

Before stating our results, let us first describe our setup in detail. Consider a system $\sys$ and an environment $\env$ described by Hilbert spaces $\hsys$ and $\henv$ respectively, which we both assume to be finite \cite{finite_comment}.  
Macroscopic constraints imposed on the system or the environment take the form of subspaces $\somg \subseteq \hsys$ and $\hil_{\Omega_E} \subseteq \henv$ respectively.
If we know, for instance, that the value of some observable $O_S$ lies within some narrow interval, $\somg$ may describe the space spanned by all eigenstates of the operator with eigenvalues within that interval.  
Before placing them into contact, the system and the environment are uncorrelated. 
That is, the initial state of $\hsys \otimes \henv$ at time $t=0$ takes the form $\ket{\phi}_S \otimes \ket{\psi}_E$,
where to explain our result we will for simplicity assume that $\ket{\phi}_S \in \somg$ and $\ket{\psi}_E \in \omg$ are pure states~\cite{mixed_states_comment}. 
The dynamics of the system and the environment, including the interactions between them, is governed by the Hamiltonian $H_{SE}$.
Given that at $t=0$ the system is in the state $\ket{\phi}_S$, it will at time $t$ be in the state
\begin{align}\label{eq:StoSchannel}
 \rho_S^\phi(t) = \tr_E \left[U(t)\left(\proj{\phi}_S\otimes\proj{\psi}_E\right)U(t)\mdag\right]\ ,
\end{align}
where $U(t) = \exp(-iH_{SE}t)$ describes the joint unitary dynamics of $S$ and $E$.
We may understand \eqref{eq:StoSchannel} as a quantum channel $\Omega_S\rightarrow S$, taking $\ket{\phi}_S$ as an input.
Its \cj representation is given by the partial trace $\tau_{\Omega_S'S}(t)=\tr_E\tau_{\Omega_S'SE}(t)$ of
\begin{align}
 \tau_{\Omega_S'SE}(t) = U(t)\left(\proj{\Psi}_{\Omega_S\Omega_S'}\otimes\proj{\psi}_E\right)U(t)\mdag\ .
\end{align}
where by definition $U(t)$ acts on $S$ and $E$, but not on $\Omega_S'$.
Applying condition \eqref{eq:cond1}, we obtain that if at any time $t$ we have $\hmineps(\Omega_S'S)_{\tau(t)} \gtrsim \hmaxeps(S)_{\tau(t)}$, almost all possible initial states $\ket{\phi}_S$ in $\somg$ (possible input states to the channel), will have evolved close to
\begin{align}
 \tau_{S}(t) &= \tr_E\tr_{\Omega_S'}\left[U(t)\left(\proj{\Psi}_{\Omega_S\Omega_S'}\otimes\proj{\psi}_E\right)U(t)\right] \nn\\
&=\tr_E\left[U(t)\left(\pi_{\Omega_S}\otimes\proj{\psi}_E\right)U(t)\right]\ .
\end{align}
This is the state $S$ would be in at time $t$ if at $t=0$ it had been maximally mixed on $\Omega_S$.
Note that $\tau_{\Omega_S'SE}(t)$ is a pure state. Thus, applying \eqref{eq:pureSym} we find that the condition $\hmineps(\Omega_S'S)_{\tau(t)} \gtrsim \hmaxeps(S)_{\tau(t)}$ is equivalent to
\begin{align}\label{eq:indepState}
 \hmaxeps(S)_{\tau_{SE}(t)} \lesssim \hmineps(E)_{\tau_{SE}(t)}\ .
\end{align}
The state $\tau_{SE}(t)$ appearing in this condition is given by
\begin{align}\label{eq:tau}
 \tau_{SE}(t) = U(t)\left(\pi_{\Omega_S}\otimes\proj{\psi}_E\right)U(t)\mdag\ ,
\end{align}
i.e., the global state of $SE$ if at $t=0$ it had been maximally mixed on $\Omega_S$.
Note that $\hmaxeps(S)_{\tau(0)}\simeq\log d_{\Omega_S}$ and $\hmineps(E)_{\tau(0)}\simeq0$ (neglecting small corrections due to smoothing), so \eqref{eq:indepState} will certainly not be fulfilled for small enough times $t$,
and it does of course depend on the details of $H_{SE}$ and $\ket{\psi}_E$ whether~\eqref{eq:indepState} can ever be satisfied at a later point in time.
Conversely, applying condition \eqref{eq:cond2} to the channel $\Omega_S\rightarrow S$ given by \eqref{eq:StoSchannel}  and that again due to purity of $\tau_{\Omega_S'SE}(t)$ and \eqref{eq:pureSym} we have $\hmaxeps(\Omega_S'S)_{\tau(t)}=\hmaxeps(E)_{\tau(t)}$ we find that as long as
\begin{align}\label{eq:notIndepState}
  \hmineps(S)_{\tau_{SE}(t)} \gtrsim \hmaxeps(E)_{\tau_{SE}(t)}
\end{align}
different initial states from $\Omega_S$ will not (yet) have evolved close to the same state.

Let us point out again that the great benefit of our result lies in the fact that it allows us to make statements about how \emph{almost all} initial states from $\Omega_S$ evolve, by analyzing entropy changes of the \emph{single} state $\tau_{SE}(t)$.
This state depends on the space $\hil_{\Omega_S}$ (i.e., the ``macroscopic constraint'') itself, the initial state of the environment $\ket{\psi}_E$, the Hamiltonian $H_{SE}$ and the time $t$ that has passed since we have put $S$ in contact with $E$,
but not on any of the individual initial states that we might place the system in.

Since we have modelled both $S$ and $E$ to be of finite dimension, the state $\tau_{SE}(t)$ will come arbitrarily close to its initial state $\tau_{SE}(0)$ in finite time, implying that condition~\eqref{eq:notIndepState} will be fulfilled: 
the system regains information about its initial state, even if it has been lost at intermediate times. 
However, for environments $E$ that consist of a macroscopic number of particles, this recurrence times will typically be very large \cite{recurrence_comment}.

We usually think of the environment to consist of much more particles than the system and hence to also be dimension-wise much larger. 
If it is, on the other hand, the case that the environment $E$ is dimension-wise sufficiently smaller than the restricted system $\Omega_S$, it will not have enough degrees of freedom to ``absorb'' all the information about the initial state of the system; hence the system will for all times retain some memory about almost any possible initial state. 
This will be the case if only a few particles are effectively interacting with a relatively large $S$.
The above intuition can be made rigorous by estimating the entropic terms in \eqref{eq:notIndepState}. Namely we have by use of the strong subbaditivity of the smooth min-entropy \cite[footnote 7]{tomamichel:duality}, the chain rule \eqref{minchain1eqn} and the definitions of the smooth entropy measures and $\tau_{SE}(t)$ that
\begin{align}
\hmineps(S)_{\tau(t)}-\hmaxeps(E)_{\tau(t)}
&\geq \hmineps(S)_{\tau(t)}-\log d_E \nn\\
&\geq \hmineps(S|E)_{\tau(t)}-\log d_E \nn\\
&\geq \hmineps(SE)_{\tau(t)}-2\log d_E \nn\\
&\geq \log d_{\Omega_S}-2\log d_E\ .
\end{align} 
Condition \eqref{eq:notIndepState} will thus be fulfilled for all times $t$ if $\log d_{\Omega_S} > 2\log d_E$.

\subsection{Example: time-scales in systems with an effective light-cone}
We can guarantee that the system still ``remembers'' its initial state as long as \eqref{eq:notIndepState} with $\tau_{SE}(t)$ as defined in \eqref{eq:tau} is fulfilled. 
It is thus interesting to study how fast $\hmineps(S)_{\tau}$ decreases from its initial value $\log d_{\Omega_S}$ and how fast $\hmaxeps(E)_{\tau}$ increases from zero. The answer to this question of course depends on the speficic model under consideration. 
In a physical model with only local interactions results of the Lieb-Robinson type like \cite{jens:entanglement} may be applied to bound the rates with which the min- and max-entropies can be changed. 

Let $S$ be a connected subset of a qubit lattice, $E$ its complement and $\Omega_S=S$. 
If $S$ is truly expanded in all spatial dimensions of the lattice, we expect for dimensional and geometrical reasons that 
\begin{align}
 \hmaxeps(E)_{\tau_{SE}(t)}\in O\left(|\partial S|\cdot v_{\text{LR}}\cdot t\right)\ ,
\end{align}
where $v_{\text{LR}}$ is the Lieb-Robinson velocity and $|\partial S|$ denotes the number of spins on the boundary of $S$ that directly interact with $E$. This can indeed be shown in a conceptually simple brute force estimate of all relevant terms applicable to non-local Hamiltonians.
Similarly, 
\begin{align}
 \log d_S - \hmineps(S)_{\tau_{SE}(t)} \in O\left(|\partial S|\cdot v_{\text{LR}}\cdot t\right)\ .
\end{align}
Since $\log d_S \sim |\partial S|\cdot\ell$, with $\ell$ the linear size of $S$, we find with criterion \eqref{eq:notIndepState} a lower bound $O\left(\ell/v_{\text{LR}}\right)$ on the time needed for different initial states of $S$ to evolve to the same state, and hence on the thermalization time of $S$.
Note that the very same lower bound has been derived in \cite{sergey:LR} for the time needed to prepare topological order in $S$, starting from a state which does not have topological order. \cite{experiment_comment}

Recent work tackled the problem of thermalization time-scales from other angles: in \cite{tony:finiteTime} sufficient time-scales for \emph{equilibration} were derived. In contrast, note that we are 
interested in \emph{necessary} time-scales for \emph{thermalization}.
Different authors have studied thermalization time-scales for Hamiltonians with randomly chosen eigenstates~\cite{randomHam1,randomHam2,randomHam3}. In contrast, our results apply to every specific Hamiltonian $H_{SE}$.

\subsection{Example: depolarizing noise}
Consider a system $S^n$ of $n$ qubits each suffering the influence of depolarizing noise. Each single-qubit state $\rho_S$ is mapped to 
\begin{align}
 (1-p)\cdot\rho_S + \sum_{i=1}^3 \frac{p}{3} \cdot\sigma_i \rho_S \sigma_i = (1-\frac{4}{3}p)\cdot\rho_S + \frac{4}{3}p\cdot\pi_S\ . 
\end{align}
Using a Stinespring dilation, this mapping can be expressed in the form $\tr_E\left[U\left(\rho_S\otimes\proj{\psi}_E\right)U\mdag\right]$ with $\dim\hil_E=4$, 
 $\ket{\psi}_E = \sqrt{1-p}\ket{0}_E+\sum_{i=1}^3\sqrt{\frac{p}{3}}\ket{i}_E$, and $U_{SE}=\sum_{\alpha=0}^3 \sigma_\alpha\otimes\proj{\alpha}_E$.
This gives rise to a state $\tau_{SE}$ as used in criterions \eqref{eq:indepState} and \eqref{eq:notIndepState}.
Since we are dealing here with an i.i.d.\ scenario, the von Neumann entropy becomes relevant for large $n$. 
There is a critical probability $p_c$ determined through $\vN(S)_\tau = \vN(E)_\tau$ above which all but exponentially few states of the $n$-qubit space $\hil_S^{\otimes n}$ are mapped to $\tau_S^{\otimes n}=\pi_S^{\otimes n}$.
We have $\vN(S)_\tau = \log2$ and $\vN(E)_\tau = \vN\left(1-p,\frac{p}{3},\frac{p}{3},\frac{p}{3}\right)$, yielding $p_c=18.93\%$ \cite{dilation_comment}. 
While only at $p_{\text{max}}=75\%$ literally all $n$-qubit states are mapped to $\pi_S^{\otimes n}$, the statement is true for almost all of them already at $p_c\approx p_{\text{max}}/4$.
Our value for $p_c$ coincides with the Hashing Bound \cite{bennett:entanglement} and is astonishingly close to the recently established threshold up to which topological codes can withstand depolarizing noise \cite{bombin:resilience}.

\section{Independence of the initial state of the environment}\label{sec:environmentState} 
We proceed to show a similar statement about the role of the initial state of the \emph{environment}. Let us thus now fix the state of the system $\ket{\phi}_S$.
We can then understand
\begin{align}\label{eq:EtoSchannel}
 \rho_S^\psi(t) = \tr_E \left[U(t)\left(\proj{\phi}_S\otimes\proj{\psi}_E\right)U(t)\mdag\right]\ ,
\end{align}
i.e.\ the dependence of the state of $S$ at time $t$ on the initial state $\ket{\psi}_E\in\hil_{\Omega_E}$, as a quantum channel $\Omega_E\rightarrow S$.
This channel's \cj representation is given by $\tilde{\tau}_{S\Omega_E'}(t)=\tr_E\tilde{\tau}_{S\Omega_E\Omega_E'}(t)$, where
\begin{align}
 \tilde{\tau}_{SE\Omega_E'}(t) = U(t)\left(\proj{\phi}_S\otimes\proj{\Psi}_{\Omega_E\Omega_E'}\right)U(t)\mdag
\end{align}
is again a pure state.
Using \eqref{eq:pureSym} we thus find that 
\begin{align}\label{eq:envSym}
 \ensuremath{{\rm H}}_{\min/\max}^\eps(S\Omega_E')_{\tilde{\tau}(t)} = \ensuremath{{\rm H}}_{\min/\max}^\eps(E)_{\tilde{\tau}(t)}\ ,
\end{align}
where the entropies on $E$ can be evaluated for the state
\begin{align}\label{eq:tauEnv}
	\tilde{\tau}_{SE}(t) = U(t) \left( \proj{\phi}_S \otimes \pi_{\Omega_E} \right) U(t)^\dagger\ .
\end{align}
This state is perfectly analogous to the state $\tau_{SE}(t)$ in \eqref{eq:tau}, which appears in our criteria \eqref{eq:indepState} and \eqref{eq:notIndepState} for independence of the initial state of the \emph{system}.
The sole difference is that the state $\tilde{\tau}_{SE}(t)$ is obtained from evolving a definite state on $S$ tensored with a maximally mixed state on some subspace $\Omega_E$ of $E$, rather than the other way round.

Applying criterion \eqref{eq:cond1} and \eqref{eq:envSym} to the channel $\Omega_E\rightarrow S$ \eqref{eq:EtoSchannel}, we find that if at time $t$
\begin{align}\label{eq:envIndep}
	\hmineps(E)_{\tilde{\tau}_{SE}(t)} \gtrsim \hmaxeps(S)_{\tilde{\tau}_{SE}(t)}\ ,
\end{align}
all but exponentially few initial states $\ket{\psi}_E\in\Omega_E$ (in the sense of Theorem~\ref{thm:main}) will yield an evolved state $\rho_S^\psi(t)$ of the system which is close to $\tilde{\tau}_{S}(t)$.
In more physical terms, the evolved state of the system then only depends on macroscopic parameters (like temperature or pressure) of the environment, but not on its precise microstate.
Since $\tilde{\tau}_{SE}(0)$ has entropy $\log d_{\Omega_E}$ in $E$ and zero entropy in $S$, condition \eqref{eq:envIndep} is certainly fulfilled for short enough times $t$.

On the other hand, applying criterion \eqref{eq:cond2} and \eqref{eq:envSym} to the channel \eqref{eq:EtoSchannel} we find that if 
\begin{align}
	\hmaxeps(E)_{\tilde{\tau}_{SE}(t)} \lesssim \hmineps(S)_{\tilde{\tau}_{SE}(t)}\ ,
\end{align}
then a substantial fraction of initial states of $\Omega_E$ lead to different states of the system at time $t$.
Our condition is again tight up to differences in min- and max-entropies of the state $\tilde{\tau}_{SE}(t)$, which vanish in the asymptotic limit.

\subsection{Example: large environment}
If the environment is very large compared to the 
system (i.e. $\log d_{\Omega_E} > 2 \log d_{S}$) then \eqref{eq:envIndep} will always be fulfilled and the 
evolved state of the system will be the same for all but exponentially few initial states of the environment restricted to $\hil_{\Omega_E}$. 
This is the situation usually encountered in physical scenarios.
Indeed, we have
\begin{align}
&\hmineps(E)_{\tilde{\tau}_{SE}(t)} - \hmaxeps(S)_{\tilde{\tau}_{SE}(t)} \nn\\
&\quad\geq \hmineps(E|S)_{\tilde{\tau}_{SE}(t)} - \hmaxeps(S)_{\tilde{\tau}_{SE}(t)} \nn\\
&\quad\geq \hmineps(ES)_{\tilde{\tau}_{SE}(t)} - \log d_S - \hmaxeps(S)_{\tilde{\tau}_{SE}(t)} \nn\\
&\quad\geq \log d_E - 2\log d_S\ . 
\end{align}
The first inequality is based on the strong subbaditivity of the smooth min-entropy \cite[footnote 7]{tomamichel:duality}, the second on the chain rule \eqref{minchain1eqn}, and the third on the definition of the smooth min- and max-entropies and the state $\tilde{\tau}_{SE}(t)$.

Formally, we obtain from Theorem~\ref{thm:main} that
\begin{align}\label{eq:largeEnv}
\Pr_{\ket{\psi}_E}\left[\trdist{\rho^\psi_S\left(t\right)}{\tilde{\tau}_S(t)}>\frac{d_S}{\sqrt{d_{\Omega_E}}}+d_{\Omega_E}^{-1/3}\right]<e^{-d_{\Omega_E}^{1/3}/16}\ ,
\end{align}
where the probability is computed over the choice of $\ket{\psi}_E$ from the Haar measure on $\hil_{\Omega_E}$.

Note that the condition ``all but exponentially few'' is not a mathematical artifact of our proof. For some examples, one can find very specific initial states of the environment that will lead to observable effects on the system even if the environment is large. 
A similar statement was shown before for \emph{almost all} times~\cite{linden:evolution}. Since our result holds for \emph{all} times, it does in particular imply said result. 

\subsection{Example: absence of thermalization}
Finally, we consider the question whether it is at all possible for the system to forget about its initial conditions. 
In \cite{linden:evolution} it is shown that the temporal average of \sys will be independent of its initial state if the relevant energy eigenstates 
of $H_{SE}$ are sufficiently entangled. Here, we prove a converse result. 
If $E$ is sufficiently larger than $S$ we know that the time-evolved state of $S$ will be close to $\tilde{\tau}_{S}(t)$, the partial trace of \eqref{eq:tauEnv}, for almost all initial states of $E$ (we consider the case $\hil_{\Omega_E}=\hil_E$ here). 
This allows us to derive sufficient conditions under which $\rho_S(t)$, the evolved state of $S$, stays close to its initial state $\rho_S(0)$ for all times and thus never thermalizes. 
Roughly, we show that if those eigenstates of the Hamiltonian which on $S$ have most overlap with the initial state $\rho_S(0)$ are not sufficiently entangled, the state of the system will remain close to its initial state for all times, for all but exponentially few initial states $\ket{\psi}_E$ of the environment.

Our result is a stronger form of a recent result of~\cite{gogolin:absence}. 
In this reference it is shown, that even for a non-integrable Hamiltonian $H_{SE}$ the system may stay close to its initial state for most times, disproving the long-held conjecture that all non-integrable systems thermalize.
The most important advantage of our result is that we can make statements about the time-evolved state of \sys (as opposed to statements about temporal averages) and do not require \sys to be small. 
Furthermore, we do not require that \emph{all} energy eigenstates of $H_{SE}$ be close to product (as in~\cite{gogolin:absence}) but only the ones which are most relevant for the particular initial state of the system.

Let us now explain our result more precisely.
Note that the energy eigenstates $\left\{\ket{E_k}_{\sys\env}\right\}_{k}$ form a basis of the product space $\hil\ssys\otimes\hil\senv$. 
Assume that we want to approximate this basis by a product basis $\left\{\ket{i}\ssys\otimes\ket{j}\senv\right\}_{i, j}$. That is, to each energy eigenstate $\ket{E_k}_{\sys\env}$ we assign the element of the product basis $\ket{i}\ssys\otimes\ket{j}\senv$ which best approximates it and assume that this correspondence is one-to-one. Let $I(i)$ denote the set of energy eigenstates which are assigned to a state of the form $\ket{i}\ssys\otimes\ket{j}\senv$, with a fixed $i$ and an arbitrary $j$. We introduce the quantity $\delta(i)$ to quantify how well the energy eigenstates in $I(i)$ are approximated by an element of the product basis,
\begin{align}
\delta(i) := \min_{\ket{E_k}\in I(i)} \max_{j=1,\ldots,d\senv} \left\{\left| \langle E_k |_{\sys\env} | i \rangle\ssys | j \rangle\senv   \right|\right\}\ .
\end{align}
Let $\rho\ssys(t)$ denote the state of the system at time $t$ and assume that its initial state was $\rho\ssys(0) = \proj{i}\ssys$. Then at any time $t$ the probability that $\rho\ssys(t)$ is further away from its initial state than $4\delta(i)\sqrt{1-\delta(i)^2}$ (in trace distance $\left\|\ldots\right\|_1$) is exponentially small. This radius is small if $\delta(i)$ is close to $1$, that is, if the enery eigenstates which on \sys are most similar to $\proj{i}\ssys$ are sufficiently close to product. The probability is computed over the choice of the initial state of the environment $\ket{\psi}\senv$.

Formally, we have with $\rho\ssys(t)=\tr_E\left[U(t)\left(\proj{i}\ssys\otimes\proj{\psi}_E\right)U(t)\mdag\right]$ that at any time $t$
\begin{align}\label{eq:absence0}
&\Pr_{\ket{\psi}_E}\left[\trdist{\rho_S(t)}{\proj{i}_S} \right.\nn\\&\qquad\left. > 4\delta(\phi)\sqrt{1-\delta(\phi)^2}+\frac{d_S}{\sqrt{d_E}}+d_E^{-1/3}\right] \nn\\& <e^{-d_E^{1/3}/16}\ ,
\end{align}
where the proability is over the choice of $\ket{\psi}_E$ from the uniform measure on $\hil_E$. 
We think of the terms $\frac{d_S}{\sqrt{d_E}}+d_E^{-1/3}$ as negligible against $4\delta(i)\sqrt{1-\delta(i)^2}$, as Hilbert space dimensions grow exponentially with the number of constituent particles  and the environment will be large in typical situations of physical interest.

In order to prove \eqref{eq:absence0}, we note that we already know from \eqref{eq:largeEnv} that if $d_E$ (recall that we have identified $E=\Omega_E$ here) is sufficiently larger than $d_S$, 
then almost all initial states $\ket{\psi}_E\in\hil_E$ of the environment will lead to the same evolved state $\tilde{\tau}_{S}(t)$ of the system, given by the partial trace of \eqref{eq:tauEnv}.
This allows us to bound 
\begin{align}\label{eq:absence1}
 \trdist{\rho_S(t)}{\proj{i}_S} \leq \trdist{\rho_S(t)}{\tilde{\tau}_{S}(t)} + \trdist{\tilde{\tau}_{S}(t)}{\proj{i}_S}\ ,
\end{align}
where the first summand on the right hand side is with high probability small due to \eqref{eq:largeEnv}. 
The explicit form of $\tilde{\tau}_{S}(t)$ allows us to bound $\trdist{\tilde{\tau}_{S}(t)}{\proj{i}_S}$, using the assumed properties of the energy eigenstates.
Namely, we find
\begin{align}\label{eq:absence2}
 \trdist{\tilde{\tau}_{S}(t)}{\proj{i}_S} \leq 4\delta(i)\sqrt{1-\delta(i)^2}\ .
\end{align}
The calculation is rather tedious and can be found in Appendix~\ref{sec:absence}. 
Combining \eqref{eq:largeEnv}, \eqref{eq:absence1}, and \eqref{eq:absence2} then yields the assertion, \eqref{eq:absence0}.

\section{Discussion}\label{sec:discussion}
We have shown that the problem of understanding how thermalizing systems become independent of their initial states can be simplified considerably -- for almost all states it suffices to understand the temporal evolution of the special state $\tau_{SE}(t)$, or rather changes in entropy for this state.
The emergence of such a special state 
is indeed somewhat analogous to the setting of channel coding, where the maximally entangled state plays an important role in quantifying a channels capacity to carry
quantum information. Note, however, that we do not ask about how much quantum information could be conveyed by using any form of coding scheme. 
Furthermore, merely asking whether the state of the system depends on its initial state after some time, or in more information theoretic terms, asking whether the output state of the channel depends on its input state does (unlike in the classical world) not immediately answer the question whether this channel is useful for transmitting quantum 
information~\cite{channel_comment}.

Note that all our statements hold ``for almost all initial states from the Haar measure'', i.e., we make statements about the volume of states. Of course, from a given starting state it is in general not the case
that all such states could be reached in a physical system, and hence one might question the relevance of our results. 
Note, however, that our approach applies to \emph{any} set of unitaries (describing the different initial states) which have such a decoupling effect.
In \cite{harrow:2design} it is shown that random two-qubit interactions efficiently approximate the first and second moments of the Haar distribution, thereby constituting \emph{approximate 2-designs}. This is all one needs for decoupling~\cite{patrick:blackHole,oleg:decoupling}. Our findings therefore apply to efficiently preparable sets of initial states. It is an interesting open question what other sets of unitaries have this property.

\section{Acknowledgements}
We gratefully acknowledge Andrew Doherty, Jens Eisert, Christian Gogolin and Tony Short for interesting discussions and comments.
This research was supported by the National Research Foundation and Ministry of Education, Singapore.

\onecolumngrid
\newpage

\appendix
\begin{center}{\bfseries APPENDIX}\end{center}

\section{Proof of Eq.~\eqref{eq:absence2}}\label{sec:absence}

In order to prove \eqref{eq:absence2}, we need a somewhat more involved notation than in the main part of this article. Consider a basis $\left\{\ket{i}_S\right\}_{i=1,\ldots,d_S}$ of $\hil_S$ and a basis $\left\{\ket{j}_E\right\}_{j=1,\ldots,d_E}$ of $\hil_E$. Both $\left\{\ket{i}_S\otimes\ket{j}_E\right\}_{i=1,\ldots,d_S , j=1,\ldots,d_E}$ and $\left\{\ket{E_k}_{SE}\right\}_{k=1,\ldots,d_S d_E}$ form bases of the Hilbert space $\hil_S\otimes\hil_E$. We consider mappings between these two bases, i.e.\ mappings of the form
\begin{align}
&\left\{1,\ldots,d_Sd_E\right\} & \longrightarrow & \quad \left\{1,\ldots,d_S\right\}\times\left\{1,\ldots,d_E\right\} \nn\\ 
&\qquad \quad k & \longmapsto & \qquad \quad \left(\xi(k),\hat{\xi}(k)\right)\ 
\end{align}
and define
\begin{align}
f_k := F\left(\ket{E_k},\ket{\xi(k)}_S\ket{\hat{\xi}(k)}_E\right)\ .
\end{align}
We are interested in how good product states of the form $\ket{\phi}_S\otimes\ket{j}_E$ (with a fixed $\phi\in\left\{1,\ldots,d_S\right\}$ and arbitrary $j$) can be  approximated by such a mapping $k\mapsto\left(\xi(k),\hat{\xi}(k)\right)$. Note that if this product states have high overlap with an energy eigenstate, this necessarily implies that the eigenstate is lowly entangled. We restrict to mappings which are injective and pick the one which maximizes $f_k$ for states of the form $\ket{\phi}_S\otimes\ket{j}_E$. Formally, we are interested in the quantity
\begin{align}\label{deltadefi}
\delta(\phi) := \max_{k\mapsto\left(\xi(k),\hat{\xi}(k)\right)} \min_{k} \left\{f_k : \xi(k)=\phi \text{\: and \: } k\mapsto\left(\xi(k),\hat{\xi}(k)\right) \text{is injective}\right\}\ .
\end{align}

We will show the following.

\begin{lemma}\label{memorythm}
With the notation introduced above, consider an initial state $\ket{\phi}_S$ with $\phi\in\left\{1,\ldots,d_S\right\}$. Assume that $\delta(\phi)>\frac{1}{\sqrt{2}}$. Then with 
\begin{align}
 \tau_S(t)=\tr_E\left[U(t)\left(\proj{\phi}_S\otimes\pi_E\right)U(t)\mdag\right]\ .
\end{align}
 we have
\begin{align}\label{memoryeqn}
\left\|\tau_S(t)-\proj{\phi}_S\right\|_1
\leq
4\delta(\phi)\sqrt{1-\delta(\phi)^2}
\end{align}for all times $t$.
\end{lemma}

By definition (\ref{deltadefi}) the requirement $\delta(\phi)>\frac{1}{\sqrt{2}}$ requires that $f_k>\frac{1}{\sqrt{2}}$ if $\xi(k)=\phi$. If this condition is fulfilled, the r.h.s.\ of (\ref{memoryeqn}) is smaller than $2$ and thus non-trivial. 

\begin{proof}
In order to shorten our notation we introduce the shorthands $\phi_S\equiv\proj{\phi}_S$, $\Psi_{EE'}\equiv\proj{\Psi}_{EE'}$, $\xi(k)_S\equiv\proj{\xi(k)}_S$ and $\hat{\xi}(k)_E\equiv\proj{\hat{\xi}(k)}_E$. Sums with summation index $k$ or $l$ go from $1$ to $d_Sd_E$ and sums with summation index $r$ go from $1$ to $d_E$. 
By use of the assumed injectivity (and hence also bijectivity) of the mapping, we have
\begin{align}\label{kroneckers}
\delta_{\xi(k),\xi(l)}\delta_{\hat{\xi}(k),\hat{\xi}(l)}=\delta_{kl}\ . 
\end{align}
This implies that $\sum_ke^{-i E_kt}\xi(k)_S\otimes\hat{\xi}(k)_E$ is a unitary, since
\begin{align}
&\left(\sum_ke^{-i E_kt}\xi(k)_S\otimes\hat{\xi}(k)_E\right)\left(\sum_le^{-i E_lt}\xi(l)_S\otimes\hat{\xi}(l)_E\right)\mdag
\nn\\&\quad=
\left(\sum_ke^{-i E_kt}\xi(k)_S\otimes\hat{\xi}(k)_E\right)\left(\sum_le^{+i E_lt}\xi(l)_S\otimes\hat{\xi}(l)_E\right)
\nn\\&\quad=
\sum_{kl}e^{-i (E_k-E_l)t}\delta_{\xi(k),\xi(l)}\delta_{\hat{\xi}(k),\hat{\xi}(l)}
\ketbra{\xi(k)_S}{\xi(l)_S}\otimes\ketbra{\hat{\xi}(k)}{\hat{\xi}(l)}_E
\nn\\&\quad=
\sum_{k}e^{-i (E_k-E_k)t}
\ketbra{\xi(k)}{\xi(k)}_S\otimes\ketbra{\hat{\xi}(k)}{\hat{\xi}(k)}_E
\nn\\&\quad=\id_{SE}\ .
\end{align}
We first show that $\tau_S(t)$ has high fidelity with the state 
\begin{align*}
\tr_E \left[ 
\left(\sum_ke^{-i E_kt}\xi(k)_S\otimes\hat{\xi}(k)_E\right) 
\left(\phi_S\otimes\pi_E\right)
\left(\sum_le^{+i E_lt}\xi(l)_S\otimes\hat{\xi}(l)_E\right)
\right]
\end{align*}
and then show that this state is identical with $\phi_S$. 
Since the fidelity can only increase under partial traces, that is, it can only decrease if we calculate it for purifications of the actual states, so
\begin{align}
F^2&\left\{
\tau_S(t) ,
\tr_E \left[ 
\left(\sum_ke^{-i E_kt}\xi(k)_S\otimes\hat{\xi}(k)_E\right) 
\left(\phi_S\otimes\pi_E\right)
\left(\sum_le^{+i E_lt}\xi(l)_S\otimes\hat{\xi}(l)_E\right)
\right]
\right\}
\nn\\ &=
F^2\left\{
\tr_E \left[
\left(\sum_ke^{-i E_kt}\proj{E_k}\right) 
\left(\phi_S\otimes\pi_E\right)
\left(\sum_le^{+i E_lt}\proj{E_l}\right)
\right]\right. , \nn\\
&\qquad\tr_E \left.\left[ 
\left(\sum_ke^{-i E_kt}\xi(k)_S\otimes\hat{\xi}(k)_E\right) 
\left(\phi_S\otimes\pi_E\right)
\left(\sum_le^{+i E_lt}\xi(l)_S\otimes\hat{\xi}(l)_E\right)
\right]
\right\}
\nn\\ &\geq
F^2\left\{
\left(\sum_ke^{-i E_kt}\proj{E_k}\right) 
\left(\phi_S\otimes\Psi_{EE'}\right)
\left(\sum_le^{+i E_lt}\proj{E_l}\right)
\right. , \nn\\
&\qquad\left.
\left(\sum_ke^{-i E_kt}\xi(k)_S\otimes\hat{\xi}(k)_E\right) 
\left(\phi_S\otimes\Psi_{EE'}\right)
\left(\sum_le^{+i E_lt}\xi(l)_S\otimes\hat{\xi}(l)_E\right)
\right\}\ .
\end{align}
Both these states are pure, so using that 
$F(\proj{\psi} ,\sigma) := \sqrt{\bra{\phi}\sigma\ket{\phi}}$
and that $\ket{\Psi}_{EE'}=\frac{1}{\sqrt{d_E}}\sum_{r}\ket{r}_E\ket{r}_{E'}$ we find
\begin{align}\label{memorylemmaaux}
F^2&\left\{
\tau_S(t) ,
\tr_E \left[ 
\left(\sum_ke^{-i E_kt}\xi(k)_S\otimes\hat{\xi}(k)_E\right) 
\left(\phi_S\otimes\pi_E\right)
\left(\sum_le^{+i E_lt}\xi(l)_S\otimes\hat{\xi}(l)_E\right)
\right]
\right\}
\nn\\ &\geq
\left|
\bra{\phi}_S\bra{\Psi}_{EE'}\left(\sum_le^{+i E_lt}\proj{E_l}_{SE}\right)
\left(\sum_ke^{-i E_kt}\xi(k)_S\otimes\hat{\xi}(k)_E\right)\ket{\phi}_S\ket{\Psi}_{EE'}
\right|^2
\nn\\ &= \left|\frac{1}{d_E}\sum_r
\bra{\phi}_S\bra{r}_E\left(\sum_le^{+i E_lt}\proj{E_l}_{SE}\right)
\left(\sum_ke^{-i E_kt}\xi(k)_S\otimes\hat{\xi}(k)_E\right)\ket{\phi}_S\ket{r}_E
\right|^2
\nn\\ &= \left|\frac{1}{d_E}\sum_r
\bra{\phi}_S\bra{r}_E\left(\sum_le^{+i E_lt}\proj{E_l}_{SE}\right)
\left(\sum_ke^{-i E_kt}\ket{\xi(k)}_S\ket{\hat{\xi}(k)}_E\right)
\delta_{\phi,\xi(k)}\delta_{r,\hat{\xi}(k)}
\right|^2
\nn\\ &= \left|\frac{1}{d_E}\sum_{kl} \delta_{\phi,\xi(k)} e^{-i (E_k-E_l)t}
\bra{\xi(k)}_S\bra{\hat{\xi}(k)}_E\proj{E_l}_{SE}
\ket{\xi(k)}_S\ket{\hat{\xi}(k)}_E
\right|^2
\nn\\ &= \left|\frac{1}{d_E}\sum_{kl} \delta_{\phi,\xi(k)} e^{-i (E_k-E_l)t}
F^2\left\{\ket{E_l}_{SE} ,
\ket{\xi(k)}_S\ket{\hat{\xi}(k)}_E
\right\}\right|^2
\end{align}
By definition (\ref{deltadefi}) the requirement $\delta(\phi)>\frac{1}{\sqrt{2}}$ requires that 
\begin{align}
F^2\left\{\ket{E_k}_{SE} ,
\ket{\xi(k)}_S\ket{\hat{\xi}(k)}_E
\right\}
>\half
\end{align} if $\delta_{\phi,\xi(k)}=1$. Since $\sum_lF^2\left\{\ket{E_l}_{SE} ,
\ket{\xi(k)}_S\ket{\hat{\xi}(k)}_E\right\}=1$, this also implies that 
\cite{sum_comment}
\begin{align}
\sum_{l : l\neq k}F^2\left\{\ket{E_l}_{SE} ,
\ket{\xi(k)}_S\ket{\hat{\xi}(k)}_E
\right\}
<\half
\end{align}
if $\delta_{\phi,\xi(k)}=1$.
We conclude that
\begin{align}
&\left|\sum_{k\neq l} \delta_{\phi,\xi(k)}e^{-i (E_k-E_l)t}  
F^2\left\{\ket{E_l}_{SE} ,
\ket{\xi(k)}_S\ket{\hat{\xi}(k)}_E
\right\}\right|
\nn\\&\quad\leq
\sum_{k\neq l}\delta_{\phi,\xi(k)}F^2\left\{\ket{E_l}_{SE} ,
\ket{\xi(k)}_S\ket{\hat{\xi}(k)}_E
\right\}
\nn\\&\quad\leq
\sum_{k=l}\delta_{\phi,\xi(k)}F^2\left\{\ket{E_l}_{SE} ,
\ket{\xi(k)}_S\ket{\hat{\xi}(k)}_E
\right\}\ .
\end{align}
We split up the sum $\sum_{kl}=\sum_{k=l}+\sum_{k\neq l}$ and use that for $a,b\in\Complex$ with $\left|a\right|\geq\left|b\right|$ we have $\left|a+b\right|\geq\left|a\right|-\left|b\right|$ to obtain from (\ref{memorylemmaaux})
\begin{align}
F^2&\left\{
\tau_S(t) ,
\tr_E \left[ 
\left(\sum_ke^{-i E_kt}\xi(k)_S\otimes\hat{\xi}(k)_E\right) 
\left(\phi_S\otimes\pi_E\right)
\left(\sum_le^{+i E_lt}\xi(l)_S\otimes\hat{\xi}(l)_E\right)
\right]
\right\}
\nn\\ &\geq
\left(
\frac{1}{d_E}\left|\sum_{k=l} \delta_{\phi,\xi(k)}e^{-i (E_k-E_l)t}  
F^2\left\{\ket{E_l}_{SE} ,
\ket{\xi(k)}_S\ket{\hat{\xi}(k)}_E
\right\}\right| \right.
\nn\\ &\qquad\quad \left. -
\frac{1}{d_E}\left|\sum_{k\neq l} \delta_{\phi,\xi(k)}e^{-i (E_k-E_l)t}  
F^2\left\{\ket{E_l}_{SE} ,
\ket{\xi(k)}_S\ket{\hat{\xi}(k)}_E
\right\}\right|
\right)^2
\nn\\ &\geq
\left(
\frac{1}{d_E}\sum_{k=l} \delta_{\phi,\xi(k)}
F^2\left\{\ket{E_l}_{SE} ,
\ket{\xi(k)}_S\ket{\hat{\xi}(k)}_E
\right\} 
\right.
\nn\\ &\qquad\quad \left. 
-
\frac{1}{d_E}\sum_{k\neq l} \delta_{\phi,\xi(k)}  
F^2\left\{\ket{E_l}_{SE} ,
\ket{\xi(k)}_S\ket{\hat{\xi}(k)}_E
\right\}
\right)^2\ .
\end{align}
Using that
\begin{align}
\sum_{l : l\neq k}F^2\left\{\ket{E_l}_{SE},\ket{\xi(k)}_S\ket{\hat{\xi}(k)}_E\right\}=
1-F^2\left\{\ket{E_k}_{SE},\ket{\xi(k)}_S\ket{\hat{\xi}(k)}_E\right\}=1-f_k^2
\end{align}
this simplifies to
\begin{align}
F^2&\left\{
\tau_S(t) ,
\tr_E \left[ 
\left(\sum_ke^{-i E_kt}\xi(k)_S\otimes\hat{\xi}(k)_E\right) 
\left(\phi_S\otimes\pi_E\right)
\left(\sum_le^{+i E_lt}\xi(l)_S\otimes\hat{\xi}(l)_E\right)
\right]
\right\}
\nn\\ &\geq \left(
\frac{1}{d_E}\sum_k \delta_{\phi,\xi(k)} 
f_k^2 -
\frac{1}{d_E}\sum_k \delta_{\phi,\xi(k)}
\left(1-f_k^2\right)
\right)^2
\nn\\ &= \left(
\frac{1}{d_E}\sum_k \delta_{\phi,\xi(k)}
\left(2f_k^2-1\right)
\right)^2
\end{align}
Applying the definition of $\delta(\phi)$ and the bijectivity of the mapping we finally obtain
\begin{align}
F^2&\left\{
\tau_S(t) ,
\tr_E \left[ 
\left(\sum_ke^{-i E_kt}\xi(k)_S\otimes\hat{\xi}(k)_E\right) 
\left(\phi_S\otimes\pi_E\right)
\left(\sum_le^{+i E_lt}\xi(l)_S\otimes\hat{\xi}(l)_E\right)
\right]
\right\}
\nn\\ &\geq
\left(\frac{1}{d_E}\left(2\delta(\phi)^2-1\right)\sum_k \delta_{\phi,\xi(k)}\right)^2 
\nn\\ &=
\left(2\delta(\phi)^2-1\right)^2\ . 
\end{align}
As for the second part of the proof,
\begin{align}\label{memorypf2}
\tr_E &\left[ 
\left(\sum_ke^{-i E_kt}\xi(k)_S\otimes\hat{\xi}(k)_E\right) 
\left(\phi_S\otimes\pi_E\right)
\left(\sum_le^{+i E_lt}\xi(l)_S\otimes\hat{\xi}(l)_E\right)
\right]
\nn\\ &=
\sum_{kl} \frac{1}{d_E} e^{-i (E_k-E_l)t} 
\delta_{\phi,\xi(k)}\delta_{\phi,\xi(l)}\delta_{\hat{\xi}(k),\hat{\xi}(l)}
\ketbra{\xi(k)}{\xi(l)}_S
\nn\\ &=
\sum_{kl} \frac{1}{d_E} e^{-i (E_k-E_l)t} 
\delta_{\phi,\xi(k)}\delta_{\xi(k),\xi(l)}\delta_{\hat{\xi}(k),\hat{\xi}(l)}\phi_S 
\end{align}
Applying (\ref{kroneckers}) for the first equality and the bijectivity for the second this simplifies to
\begin{align}
\tr_E &\left[ 
\left(\sum_ke^{-i E_kt}\xi(k)_S\otimes\hat{\xi}(k)_E\right) 
\left(\phi_S\otimes\pi_E\right)
\left(\sum_le^{+i E_lt}\xi(l)_S\otimes\hat{\xi}(l)_E\right)
\right]
\nn\\ &=
\sum_{k} \frac{1}{d_E} 
\delta_{\phi,\xi(k)}\phi_S 
\nn\\ &=
\phi_S\ .
\end{align}
We find
\begin{align}\
F\left\{
\tau_S(t) ,
\phi_S
\right\}
\geq
2\delta(\phi)^2-1 
\end{align}
and by use of the Fuchs-van de Graaf inequalities \cite{fvdg_comment,fuchs:thesis}
\begin{align}
\left\|\tau_S(t)-\phi_S\right\|_1
&\leq
2\sqrt{1-F(\tau_S(t),\phi_S)^2}
\nn\\&\leq
2\sqrt{1-\left(2\delta(\phi)^2-1\right)^2}
\nn\\&=
4\delta(\phi)\sqrt{1-\delta(\phi)^2}
\end{align}
which is lower than $2$ if $\delta(\phi)$ is larger than $\frac{1}{\sqrt{2}}$.
\end{proof}


\begin{thebibliography}{37}
\bibitem{definition_comment}
Note that some of the references like \cite{renato:diss,marco:fqaep,dupuis:diss} use older definitions for the smoothing procedure and relations found therein may therefore not be valid with the definition used in this article, which is based on \cite{tomamichel:duality}.
\bibitem{helstrom_comment}
As the trace distance quantifies how well we can distinguish $\rho$ from $\sigma$ when given with equal probability~\cite{helstrom:detection},
this says that there exists no physical process that can easily
tell them apart.
\bibitem{finite_comment}
In principal, the dimension of $E$ may be infinite. 
However, here we consider Hilbert spaces of very large but finite dimension. We believe that this restriction is not of any fundamental significance. 
Technically, it may be achieved by imposing a high-energy cutoff on an infinite-dimensional $\henv$.
\bibitem{mixed_states_comment}
Our results readily generalize to mixed initial states and even states which are correlated/entangled to some outside reference $R$. Furthermore, our results generalize to weakly coupled initial states, i.e.\ initial states that are close to a product state. If two initial state are globally close, then evolving them unitarily and tracing out \env can never increase their distance, so that all of our predictions hold with a small additional error also for states that initially are close to product.
\bibitem{experiment_comment}
An effective light-cone as described by the Lieb-Robinson bounds has
recently been measured for the first time in an experiment \cite{cheneau:lightcone}. The
time-scales we have found are thus in principal open to experimental verification.
Furthermore, using our result the found dependencies of $v_{\text{LR}}$ on microscopic
physical parameters directly translate into dependencies of thermalization
time-scales on those same parameters.
\bibitem{dilation_comment}
While the Stinespring dilation is not unique, all dilations are connected by isometries on the $E$-system and the entropies are invariant under isometries. The value found for $p_c$ is therefore independent of what particular dilation is chosen.
\bibitem{channel_comment}
Of course, making statements for almost all states does tell us that it cannot be too useful.
\bibitem{hmax_comment}
Note that the quantity called $\hmax$ in \cite{renato:diss} differs from ours. $\hmin(A|B)_\rho \geq \hmin (AB)_\rho - \hmax (B)_\rho$ does not hold.
\bibitem{sum_comment}
With $\sum_{l : l\neq k}$ we denote a sum over all values of $l$ which are not equal to $k$. With $\sum_{k\neq l}$ we denote a sum over all pairs of possible values of $k$ and $l$ which are not equal.
\end{thebibliography}
\end{document}